\documentclass[conference,onecolumn]{IEEEtran}

\usepackage[margin=2.5cm]{geometry}

\usepackage{cite} % automatically sort and range citations numbers
\usepackage{graphicx} % support graphics in eps format
\usepackage[cmex10]{amsmath} % ensure only type 1 fonts are used
\usepackage{amsfonts,amsbsy,amssymb,amsthm}
\usepackage{array}
\usepackage{color}
\usepackage{exscale}
\usepackage{multirow}

\usepackage{verbatim}

\newtheorem{thm}{Theorem}%[section]
\newtheorem{cor}{Corollary}
\newtheorem{lem}{Lemma}

\newtheorem{dfn}{Definition}

\newtheorem{sublemma}{Sublemma}
\newtheorem{identity}{Identity}
\def\QED{\mbox{\rule[0pt]{1.5ex}{1.5ex}}}
\def\endproof{\hspace*{\fill}~\QED\par\endtrivlist\unskip\vspace{1ex}}

\begin{document}

\sloppy

%% Paper Title
%% You can use linebreaks \\ within to get better formatting as
%% desired.
\title{Multi-Resolution Video Streaming in Peer-to-Peer Networks}

%% Author names and affiliations:
%%
%% Avoiding spaces at the end of the author lines is not a problem with
%% conference papers because we don't use \thanks or \IEEEmembership.
%%
%% For several authors with only one affiliation:
%%
% \author{
%   \IEEEauthorblockN{Hui-Ting Chang and Stefan M.~Moser}
%   \IEEEauthorblockA{Department of Electrical and Computer Engineering\\
%     National Chiao Tung University (NCTU)\\
%     Hsinchu, Taiwan\\
%     Email: \{email-of-hui-ting,email-of-stefan\}@ieee.org}
% }
%%
%% For up to three affiliations:
%%
\author{
\IEEEauthorblockN{
Batuhan Karag\"oz\IEEEauthorrefmark{1},
Semih Yavuz\IEEEauthorrefmark{2},
Tracey Ho\IEEEauthorrefmark{3}, and
Michelle Effros\IEEEauthorrefmark{3}}
\IEEEauthorblockA{\IEEEauthorrefmark{1}%
Department of Computer Engineering,
Middle East Technical University,
Ankara 06800, Turkey}
\IEEEauthorblockA{\IEEEauthorrefmark{2}%
Department of Mathematics,
Bilkent University,
Ankara 06800, Turkey}
\IEEEauthorblockA{\IEEEauthorrefmark{3}%
Department of Electrical Engineering,
California Institute of Technology,
Pasadena, California 91125, USA}
\IEEEauthorblockA{%
\texttt{batu@ceng.metu.edu.tr},
\texttt{y\_semih@ug.bilkent.edu.tr},
\texttt{tho@caltech.edu,}}
\texttt{effros@caltech.edu}

}
%%
%% For over three affiliations, or if they all won't fit within the width
%% of the page, use this alternative format:
%%
% \author{
%   \IEEEauthorblockN{
%     Michael Shell\IEEEauthorrefmark{1},
%     Homer Simpson\IEEEauthorrefmark{2},
%     James Kirk\IEEEauthorrefmark{3},
%     Montgomery Scott\IEEEauthorrefmark{3} and
%     Eldon Tyrell\IEEEauthorrefmark{4}}
%   \IEEEauthorblockA{
%     \IEEEauthorrefmark{1}School of Electrical and Computer Engineering\\
%     Georgia Institute of Technology, Atlanta, Georgia 30332--0250\\
%     Email: see http://www.michaelshell.org/contact.html}
%   \IEEEauthorblockA{
%     \IEEEauthorrefmark{2}Twentieth Century Fox, Springfield, USA\\
%     Email: homer@thesimpsons.com}
%   \IEEEauthorblockA{
%     \IEEEauthorrefmark{3}Starfleet Academy, San Francisco, California 96678-2391\\
%     Telephone: (800) 555--1212, Fax: (888) 555--1212}
%   \IEEEauthorblockA{
%     \IEEEauthorrefmark{4}Tyrell Inc., 123 Replicant Street, Los Angeles, California 90210--4321}
% }

%% Use for special paper notices

%% To balance the two columns, you should reduce the text-height of
%% the last page using the following command:
%%%%%%%%%%%%%%%%%%%%%%%%%%%%%%%%%%%%%%%%%%%%%%%%%%%%%%%%%%%%%%%%%%%%%
%\addtolength{\textheight}{-9.35cm}
%%%%%%%%%%%%%%%%%%%%%%%%%%%%%%%%%%%%%%%%%%%%%%%%%%%%%%%%%%%%%%%%%%%%%
%% with an appropriate value. This command must be place on the second
%% last page, i.e., for a one-page abstract here, for a two-page
%% abstract right after the \maketitle command.

%% Create the title:
\maketitle

%% Abstract:
%% For the final version of the accepted paper, please make sure you
%% remove the comment "THIS PAPER IS ELIGIBLE FOR THE STUDENT PAPER
%% AWARD."
%%
\begin{abstract}
We consider multi-resolution streaming in fully-connected peer-to-peer networks, where transmission rates are constrained by arbitrarily specified upload capacities of the source and peers. We fully characterize the capacity region of rate vectors achievable with arbitrary coding, where an achievable rate vector describes a vector of throughputs of the different resolutions that can be supported by the network. We then prove that all rate vectors in the capacity region can be achieved using pure routing strategies. This shows that coding has no capacity advantage over routing in this scenario.
\end{abstract}

\section{Introduction}\label{intro}
We consider  multi-resolution streaming in a heterogeneous peer-to-peer setting, where peers have different upload capacities and demand an information stream at different resolutions. The information stream is layered, such as in Scalable Video Coding~\cite{schwarz07}, which generates a base video layer and a number of enhancement layers that depend on the base layer and all lower layers.

We assume a fully-connected overlay network in which transmission rates are constrained by the upload capacity of the source and each peer, a model introduced in~\cite{mundinger} to capture the most important constraints in peer-to-peer networks. A problem instance is defined by specifying the number of layers demanded by each peer and the upload capacity constraints of the source and each peer. Our goal is to find the capacity region of achievable rate vectors, where an achievable rate vector describes a vector of throughputs of the different resolutions that can be supported by the network.

Solutions can be classified as follows. Inter-session coding solutions are the most general, allowing coding across information from different sessions (i.e.~layers).  Intra-session
coding solutions restrict coding to occur only within each session. Routing solutions allow only replication and forwarding of information at each node. Intra-session coding corresponds to independent multicast network coding for each layer, for which the capacity region is given by a linear program. In contrast, characterizing inter-session coding capacity, which corresponds to the  information theoretic capacity, is open for general networks.

Related work by Chiu et al.~\cite{chiu} studies the special case of a single resolution. That case corresponds to a single multicast, and~\cite{chiu} shows that network coding is not needed to achieve capacity. In~\cite{ponec}, Ponec et al. consider the multi-resolution case restricted to intra-session coding, showing that intra-session coding does not improve the capacity region over routing. A different objective of minimizing average finish times for file download was studied in~\cite{ezovski,chang}.

In this paper we provide a complete characterization of the capacity region of feasible rate vectors  achievable with arbitrary (inter or intra-session) coding, and show that the entire capacity region can be achieved with routing.

%\section{Related Work}\label{RW}

\section{Problem Definition}\label{probdef}

  A peer to peer network is modeled as a complete directed graph with a single source node $p_0$ and $k\ge 1$ peer nodes $\lbrace p_0,p_1,\ldots,p_k\rbrace $. The upload capacities of nodes $p_0,\ldots,p_k$ are $C_0,\ldots,C_k$ respectively.

Information originates at the source node and is distributed to the peers, which help the distribution process by uploading information to other peers. Coding may occur at the source and peers.

 Let $n$ be the number of different resolutions in a layered data stream. We denote by $x_1,\ldots,x_n$ the data streams corresponding to the different layers, such that the $j$th resolution corresponds to $\{x_1,x_2,\ldots, x_j\}$. %Hence, $x_j$ can be considered as a unique identifier data segment for the $j$th resolution in question. Length of a data segment $x_j$ will be
 The rate of $x_j$ is denoted by $L_j$. For simplicity, we assume that the upload capacities $C_0,\ldots,C_k$ and the data rates $L_1,\ldots,L_n$ are integers, which can be approached arbitrarily closely by scaling the unit appropriately.

We are given nested sets $X_1,\ldots,X_n$ specifying the demands:
 $$X_j=\lbrace p_i|p_i\ demands\ x_j\rbrace.$$ We also define $X_{n+1}=\lbrace p_0 \rbrace$, so we have
$$\lbrace p_0 \rbrace=X_{n+1}\subseteq X_{n}\subseteq\ldots\subseteq X_1=\lbrace p_0,p_1,\ldots,p_k\rbrace.$$

%Now we can state properties of a successful transmission scheme and definitional constraints on it.These properties will be stated in terms of information carried through the links, more precisely the information carried through the links, of network. In this paper, when we mentioned about a link, it is actually meant to be the information carried through that link. For this purpose, we will use the following notation:
For all $S_1,S_2\subseteq X_1$, $S_1\rightarrow S_2$ is defined as the set of all links coming from $S_1$ and going to $S_2$. We also write $p_i\rightarrow p_j$ instead of $\lbrace p_i\rbrace\rightarrow\lbrace p_j\rbrace$ for brevity. The constraints on a successful transmission scheme are as follows:

\textit{i)} Each outgoing link of the source $p_0$ is a function of $x_1,\ldots,x_n$:
$$\forall p_i\in X_1,\ H(p_0\rightarrow p_i|x_1,\ldots,x_n)=0.$$

\textit{ii)} Each outgoing link of $p_i\in X_1\setminus\lbrace p_0\rbrace$ is a function of incoming links:
$$\forall p_i,p_j\in X_1\setminus\lbrace p_0\rbrace,\ H(p_i\rightarrow p_j|p_0\rightarrow p_i,\ldots,p_n\rightarrow p_i)=0.$$
We assume that $H(p_i\rightarrow p_i)=0$ without loss of optimality.

\textit{iii)} Each peer $p_i\in X_1$ can transmit at rate at most $C_i$:
$$\forall p_i\in X_1\ C_i\ge \sum_{p_j\in X_1} H(p_i\rightarrow p_j)$$

\textit{iv)} Each peer $p_i\in X_j\setminus\lbrace p_0\rbrace$ is able to decode $x_j$ from its received information:
$$I(X_1\rightarrow p_i;x_j)=H(x_j)=L_j.$$

\section{Approach}
In this section we provide some intuition for our approach. A first observation is that total upload capacity should be greater than the total rate of data which has to be delivered:
	\begin{equation}\label{main}
				\sum_{i=0} C_i \ge \sum_{i=1}^n |X_i|L_i.
	\end{equation}
	This condition is necessary but not sufficient. The following sequence of lemmas leads to a sufficient condition for a rate vector to be achievable.  The proofs of lemmas in this and the next section can be found in the Appendix.	
%This leads us to investigate the inequalities given in Lemma \ref{goodvec}. Lemma \ref{routree} is the basis for Lemma \ref{1res} and Lemma \ref{1res} is the basis for Lemma \ref{goodvec}.

		\begin{lem}\label{routree}
		Let $k$ and $C_0$ be positive integers and $C_1,C_2,\ldots,C_k$ be nonnegative integers such that
		$$ C_0+\sum_{i=1}^k C_i = k $$
		Then there exists a directed tree rooted at $v_0$ with vertices $v_1,\ldots,v_k$ such that
		$$ outdeg(v_0) = C_0,$$
		$$ \forall i \in \lbrace 1,\ldots,k \rbrace\ outdeg (v_i)=C_i, \ indeg (v_i)=1. $$
		\end{lem}

		\begin{lem} \label {1res} Data $x$ with rate $L$ can be transmitted to peers $p_1,\ldots,p_k$ by using source capacity $C_0$ and peer upload capacities $C_1,\ldots,C_k$  if
		
		$$ C_0 \ge L\ and $$
		$$ \sum_{i=0}^k C_i \ge kL. $$

		\end{lem}

	\begin{lem}\label{goodvec}
	Given the sets of peers $X_1,X_2,\ldots,X_n$ and upload capacities $C_0,C_1,\ldots,C_k$, the rate vector $(L_1,L_2,\ldots,L_n)$
	is achievable if for every $j\in\lbrace 1,\ldots,n\rbrace$
	\begin{equation}\label{gv}
	\sum_{p_i\in X_j}C_i\ge \sum_{i=1}^{j-1}L_i+\sum_{i=j}^n |X_i|L_i	
	\end{equation}
	and
	\begin{equation}\label{gv2}
	C_0\ge\sum_{i=1}^n L_i.
	\end{equation}

	\end{lem}
%	\begin{proof}
%	Proof of this theorem can be found in the extended paper~\cite{Long Version}.	
%	\end{proof}
	
Intuitively, if one of the inequalities in Lemma \ref{goodvec}, say the $j$th one, does not hold, this means that the nodes in set $X_j$ cannot handle the transmission of data layers $x_j$ through $x_n$. Hence, some peers from the set $X_1\setminus X_j$ need to help in transmitting those data layers, necessitating some additional capacity for transmitting this data to peers in $X_1\setminus X_j$ which do not themselves demand it. This requires additional capacity beyond that given in (\ref{main}).

To characterize this explicitly, it is useful to define the margin of the $j$th inequality:
		$$N_j=\sum_{i=j}^n |X_i|L_i+\sum_{i=1}^{j-1}L_i-\sum_{p_i\in X_j}C_i.$$
		For completeness we also define the $(n+1)$-th margin $N_{n+1}$ as zero. The capacity region derived in the next section is stated in terms of these margins.
In fact, not all of them, but a special subset of them, will be used. This subset is defined as follows:
	\begin{dfn}
		 			For a finite sequence $\lbrace a_n \rbrace=a_1,\ldots,a_s$, the dominant subsequence of $\lbrace a_n \rbrace$ is the subsequence $\lbrace a_{i_n}\rbrace= a_{i_1},\ldots,a_{i_h}$ defined by
		 			
		 			i) $i_h=s$
		 			
		 			ii) $i_j$ is the greatest index such that $i_j<i_{j+1}$ and $a_{i_j}>a_{i_{j+1}}$.
	\end{dfn}

\section{Converse Bound for Capacity Region}
	In this section, we present a converse bound on the capacity region, which is shown to be tight in the following section. %Owing to space constraints, the proofs of subsidiary lemmas are given in the extended version of this paper~\cite{Long Version}.

	 	\begin{thm}\label{general bound}
	Given the sets of peers $X_1,X_2,\ldots,X_n$ and upload capacities $C_0,C_1,\ldots,C_k$, if the rate vector $(L_1,L_2,\ldots,L_n)$
	 		is achievable by any coding scheme, then
	 		$$\sum_{p_i\in X_1} C_i \ge \sum_{i=1}^n |X_i|L_i+\sum_{i=1}^h \frac{N_{d_i}-N_{d_{i+1}}}{|X_{d_i}|-1}$$
	 		where $N_{d_1},\ldots,N_{d_{h+1}}$ is the dominant subsequence of $N_1,\ldots,N_{n+1}$.
	 	\end{thm}
	 		\begin{proof}
	 		For a resolution $x_j$ and a peer $p_i\in X_j$ we have, from property iv in Section~\ref{probdef},
	 		$$H(x_j)=I(X_1\rightarrow p_i;x_j)=I(X_1\setminus X_j \rightarrow p_i,X_j\rightarrow p_i;x_j).$$
	 		If we view $X_1\setminus X_j$ as a supernode, outgoing links should be functions of incoming links, since peers in the set $X_1\setminus X_j$ do not create additional data besides incoming data (property ii). Hence, links in set $X_1\setminus X_j\rightarrow p_i$ are completely dependent on links in set $X_j\rightarrow X_1\setminus X_j$. Then, we may write:
	 		$$ H(x_j)=I((X_1\setminus X_j) \rightarrow p_i,X_j\rightarrow p_i;x_j)\le$$
	 		$$  I(X_j\rightarrow (X_1\setminus X_j),X_j\rightarrow p_i;x_j)\le H(x_j)  $$
	 		$$\Rightarrow  \;\; H(X_j) = I(X_j\rightarrow (X_1\setminus X_j);x_j)$$
	 		$$\hspace{0.2 in}+I(X_j\rightarrow p_i;x_j|X_j\rightarrow (X_1\setminus X_j)).$$
	 		 Summing this for all peers in $X_j$ yields
	 		 $$|X_j|H(X_j)= |X_j|I(X_j\rightarrow (X_1\setminus X_j);x_j)+$$
	 		 $$\sum_{p_i\in X_j}I(X_j\rightarrow p_i;x_j|X_j\rightarrow (X_1\setminus X_j)).$$
	 		 By rearranging this, we can obtain
	 		 $$ I(X_j\rightarrow (X_1\setminus X_j);x_j) =\frac{1}{|X_j|-1}[|X_j|H(x_j)-  $$
	 		 $$ I(X_j\rightarrow (X_1\setminus X_j);x_j)- $$
	 		 \begin{equation}\label{1}
	 		 \sum_{p_i\in X_j}I(X_j\rightarrow p_i;x_j|X_j\rightarrow (X_1\setminus X_j))].
	 		 \end{equation}
	 		 The left hand side of this equation can be replaced by parameters which are independent from the transmission scheme by using the following lemma:
	 		 \begin{lem}\label{lhs}
	 		 	 		$ \sum\limits_{j=1}^{n}I(X_j\rightarrow (X_1\setminus X_j);x_j)\le\sum_{p_i\in X_j}C_i - \sum_{j=1}^{n} |X_j|H(x_j).$
	 		 	 	\end{lem}

	 		 	 	By using  Equation (\ref{1}) and Lemma \ref{lhs}, we can obtain:
	 		 	 	$$\sum_{p_i \in X_j}C_i\ge \sum_{j=1}^{n} |X_j|H(x_j)+\sum_{j=1}^{n} \frac{1}{|X_j|-1}[	|X_j|H(x_j) $$
	 		 	 	$$-I(X_j\rightarrow (X_1\setminus X_j);x_j) $$
	 		 	 	\begin{equation}\label{2}
	 		 	  -\sum_{p_i\in X_j}I(X_j\rightarrow p_i;x_j|X_j\rightarrow (X_1\setminus X_j))].
	 		 	 	\end{equation}	
	 		 	 		Now define $A_j$ as
	 		 	 		$$A_j=\sum_{l=j}^n|X_l|H(x_l)-\sum_{l=j}^{n}[ I(X_l\rightarrow (X_1\setminus X_l);x_l)+$$
	 		 	 		$$\sum_{p_i\in X_l}I(X_l\rightarrow p_i;x_l|X_l\rightarrow (X_1\setminus X_l))].$$
	 		 	 		For the sake of completeness, define also $A_{n+1}=0$.
	 		 	 		Putting this into Inequality (\ref{2}), we  obtain:
	 		 	 	\begin{equation}\label{3}
	 		 	 	\sum_{p_i \in X_j}C_i\ge  \sum_{j=1}^{n} |X_j|H(x_j)+\sum_{j=1}^{n} \frac{A_j-A_{j+1}}{|X_j|-1}.
	 		 	 	\end{equation}	 		 	 		
	 		 	 		We also have, from Equation (\ref{1}):
	 		 	 		\begin{equation}\label{positivity}
	 		 	 		A_j-A_{j+1}=(|X_j|-1) I(X_j\rightarrow (X_1\setminus X_j);x_j)\ge 0.
	 		 	 		\end{equation}
	 		 	 		
	 		 	 		Note that the $A_j$ values are determined by the transmission scheme. To obtain a bound which is independent from transmission scheme, we use the following lemma:
	 		 	 		 		\begin{lem}\label{N_j'}
	 		 	 		 		$A_j\ge N_j$.
	 		 	 		 		\end{lem}

	 			Let us examine the last sum in (\ref{3}):
	 			$$\sum_{j=1}^n \frac{A_j-A_{j+1}}{|X_j|-1}=\sum_{j=1}^{d_1-1}\frac{A_j-A_{j+1}}{|X_j|-1}+\sum_{j=d_1}^{d_2-1}\frac{A_j-A_{j+1}}{|X_j|-1}+\ldots+$$
	 			$$\sum_{j=d_h}^n\frac{A_j-A_{j+1}}{|X_j|-1}.$$
	 			Using inequality (\ref{positivity}),\begin{eqnarray*}
	 			& &\hspace{-0.27 in}0+\sum_{j=d_1}^{d_2-1}\frac{A_j-A_{j+1}}{|X_{d_1}|-1}+\ldots+\sum_{j=d_h}^n\frac{A_j-A_{j+1}}{|X_{d_h}|-1}\\
	 			& = &	\frac{A_{d_1}}{|X_{d_1}|-1}+A_{d_2}\left(\frac{1}{|X_{d_2}|-1}-\frac{1}{|X_{d_1}|-1}\right)+\ldots+\\
	 			& &A_{d_h}\left(\frac{1}{|X_{d_h}|-1}-\frac{1}{|X_{d_{h-1}}|-1}\right)\\&\ge& \frac{N_{d_1}}{|X_{d_1}|-1}+
	 			N_{d_2}\left(\frac{1}{|X_{d_2}|-1}-\frac{1}{|X_{d_1}|-1}\right)+\ldots+\\
	 			& &N_{d_h}\left(\frac{1}{|X_{d_h}|-1}-\frac{1}{|X_{d_{h-1}}|-1}\right)\\& = &\sum_{i=1}^h \frac{N_{d_i}-N_{d_{i+1}}}{|X_{d_i}|-1}.\end{eqnarray*}
	 			The last  inequality is due to Lemma \ref{N_j'} . By putting this result in (\ref{3}), we obtain:
	 			$$\sum_{p_i\in X_1} C_{i} \ge \sum_{i=1}^n |X_i|L_i+\sum_{i=1}^h \frac{N_{d_i}-N_{d_{i+1}}}{|X_{d_i}|-1}$$
	 		\end{proof}

\section{A Routing Scheme that Achieves the Capacity  Region}\label{method}
In this section, we give a transmission scheme using multicast routing trees that achieves the bound in Theorem \ref{general bound}.

	\begin{thm}\label{final}
	Given the sets of peers $X_1,X_2,\ldots,X_n$ and upload capacities $C_0,C_1,\ldots,C_k$, the rate vector $(L_1,L_2,\ldots,L_n)$
	is achievable by routing if 		
		\begin{equation}\label{final1}
			C_0\ge \sum_{i=1}^n L_i,
		\end{equation}
		\begin{equation}\label{final2}
			\sum_{p_i\in X_1} C_i \ge \sum_{i=1}^n |X_i|L_i+\sum_{i=1}^h \frac{N_{d_i}-N_{d_{i+1}}}{|X_{d_i}|-1}
		\end{equation}
		where $N_{d_1},\ldots,N_{d_{h+1}}$ is the dominant subsequence of $N_1,\ldots,N_{n+1}$.
	\end{thm}
	\begin{proof}
	The proof is by induction on the number of inequalities from Lemma \ref{goodvec} which are not satisfied. For this purpose let us define the set
	$$I=\lbrace i|N_i>0\rbrace.$$
	We will show that we can reduce the size of $I$ by at least one, by using some amount of capacity, such that the residual system also satisfies (\ref{final1}) and (\ref{final2}).
	
	The base case is the one where $I$ is the empty set, i.e. all $N_j$ values are less than or equal to zero and it is examined in Lemma \ref{goodvec}.
	Let $m$ and $M$ denote the minimum and maximum elements of $I$, respectively. The following lemma is essential for our method:
	
	\begin{lem}\label{Choice of C_iM}
			$\forall p_i \in X_1\setminus X_m, \ \exists C_{iM}$ with $0 \le C_{iM} \le C_i$ such that
			$$\sum_{p_i \in X_j\setminus X_m} C_{iM} \le N-N_j$$ $\forall j\in\{1, 2, \ldots, m-1\}$ and
			$$\sum_{p_i \in X_1\setminus X_m} C_{iM} = \frac{|X_M|N}{|X_M|-1}$$ where $N= min(N_m,N_M)$.
	\end{lem}

	Now, for each $p_i\in X_1\setminus X_m$ let us choose $C_{iM}$ as in Lemma \ref{Choice of C_iM}. Then we have
	\begin{equation}\label{helpamount}
	\sum_{p_i \in X_1\setminus X_m} C_{iM} = \frac{|X_M|N}{|X_M|-1}
	\end{equation}
	where $N=\min (N_m,N_M)$. Let us also define
	$$C_{0M}=\frac{N}{|X_M|-1}.$$
	Now take a rate-$\frac{N}{|X_M|-1}$ portion of $x_M$, called $s$. Then, using equality (\ref{helpamount}), the rate of $s$ is given by
	$$\frac{N}{|X_M|-1}=\sum_{p_i \in X_1\setminus X_m} \frac{C_{iM}}{|X_M|}.$$
	Hence we can divide $s$ into $|X_1\setminus X_m|$ portions $s_i$ corresponding to peers $p_i\in X_1\setminus X_m$, where the rate of portion $s_i$ is given by $\frac{C_{iM}}{|X_M|}$. To each peer $p_i\in X_1\setminus X_m$, send $s_i$ from the source. This consumes $C_{0M}$ amount of capacity of the source. Then, from each $p_i\in X_1\setminus X_m$, send $s_i$ to the peers in $X_M$. This consumes $C_{iM}$ amount of capacity of each $p_i\in X_1\setminus X_m$. In this way, transmission of portion $s$ is completed. After the procedure, we have residual capacities
		$$
						 C_i'=
						\left\{
						\begin{array}{cl}
						 C_i-C_{iM} & if \ p_i\in X_1\setminus X_m \\
						 C_i & if \ p_i \in X_m\setminus \lbrace p_0\rbrace \\
						 C_0-\frac{N}{|X_M|-1} & if\ p_i=p_0.
						\end{array}
						\right.
		$$
	and residual data rates
					$$
						 L_i'=
						\left\{
						\begin{array}{cl}
						 L_i-\frac{N}{|X_M|-1} & if\ i=M \\
						 L_i & otherwise.
						\end{array}
						\right.
					$$
	 The $N_i$ values are updated accordingly. Denoting the updated value of $N_j$ by $N_j'$, we calculate it differently for three cases:
	
	\textit{i)} If $j<m$:
		\begin{eqnarray*}N_j'& = &\sum_{i=j}^n |X_i|L_i'-C_0'+\sum_{i=1}^{j-1}L_i'-\sum_{p_i\in X_j\setminus \lbrace p_0\rbrace}C_i'\\
		& = &\sum_{i=j}^n |X_i|L_i-\frac{|X_M|N}{|X_M|-1}-C_0+\frac{N}{|X_M|-1}+\sum_{i=1}^{j-1}L_i\\
		& &-\sum_{p_i\in X_j\setminus \lbrace p_0\rbrace}C_i+\sum_{p_i\in X_j\setminus X_m}C_{iM}\\& = &(\sum_{i=j}^n |X_i|L_i+\sum_{i=1}^{j-1}L_i-\sum_{p_i\in X_j}C_i)\\
		& &-(\frac{|X_M|N}{|X_M|-1}-\frac{N}{|X_M|-1})+\sum_{p_i\in X_j\setminus X_m}C_{iM}\\
		& = &N_j-N+\sum_{p_i\in X_j\setminus X_m}C_{iM}.\end{eqnarray*}
		By the choice of $C_{iM}$ values, using Lemma \ref{Choice of C_iM}, we have
		\begin{equation}\label{i<m}
		N_j'= N_j-N+\sum_{p_i\in X_j\setminus X_m}C_{iM}\le 0
	\end{equation}
	
	\textit{ii)} If $m\le j\le M$:
		\begin{eqnarray*}N_j'& = &\sum_{i=j}^n |X_i|L_i'-C_0'+\sum_{i=1}^{j-1}L_i'-\sum_{p_i\in X_j\setminus \lbrace p_0\rbrace}C_i'\\
		& = &\sum_{i=j}^n |X_i|L_i-\frac{|X_M|N}{|X_M|-1}-C_0+\frac{N}{|X_M|-1}\\
		& & +\sum_{i=1}^{j-1}L_i-\sum_{p_i\in X_j\setminus \lbrace p_0\rbrace}C_i\end{eqnarray*}
		\begin{equation}\label{m<i<M}\hspace{- 2 in}
		\Rightarrow\;\;N_j'= N_j-N.
		\end{equation}
		
	\textit{iii)} If $j>M$:
		\begin{eqnarray*}N_j'& = &\sum_{i=j}^n |X_i|L_i'-C_0'+\sum_{i=1}^{j-1}L_i'-\sum_{p_i\in X_j\setminus \lbrace p_0\rbrace}C_i'\\
		& = &\sum_{i=j}^n |X_i|L_i-C_0+\frac{N}{|X_M|-1}+\sum_{i=1}^{j-1}L_i\\
		& &-\frac{N}{|X_M|-1}-\sum_{p_i\in X_j\setminus \lbrace p_0\rbrace}C_i\end{eqnarray*}
		\begin{equation}\label{i>M}\hspace{- 2 in}
		\Rightarrow\;\;N_j'= N_j<0.
		\end{equation}
	
	Now let us examine the dominant subsequence $N_{d_1'}',\ldots,N_{d_{h'+1}'}'$ of $N_1',\ldots,N_{n+1}'$.
	%$$N_{d_1'}',\ldots,N_{d_{h'+1}'}'.$$
	Since $N_{n+1}'$ is zero by definition, $N_{d_{h'+1}'}'$ is also zero. Therefore for all $ i\in \lbrace 1,\ldots,h'\rbrace$
	$$N_{d_{i}'}'>N_{d_{h'+1}'}'=0 \Rightarrow m\le d_{i}' \le M.$$
	But the order of $N_i'$ values for $ m\le i \le M $ is the same as the order of $N_i$ values for $ m\le i \le M $ since $N_j'= N_j-N$. Also we know that for all $ i\in \lbrace 1,\ldots,h\rbrace\ $
	$$m\le d_i\le M.$$
	Noting that $N_{d_h}=N_M$, we have two cases:
	$$(N_{d_1'}',\ldots,N_{d_{h'+1}'}')=(N_{d_1}-N,\ldots,N_{d_h}-N,0)$$
	 if $N=N_m<N_M$, and
	$$(N_{d_1'}',\ldots,N_{d_{h'+1}'}')=(N_{d_1}-N,\ldots,N_{d_h-1}-N,0)$$
	if $N=N_M\le N_m$.
	~\\
	These two cases can be considered as one since even if $N$ is equal to $N_M$, we can consider as $N_{d_{h'}'}'=0=N_{d_{h}}-N$ so that it does not affect inequality (\ref{final2}). Hence we can write:
	$$(N_{d_1'}',\ldots,N_{d_{h'}'}',N_{d_{h'+1}'}')=(N_{d_1}-N,\ldots,N_{d_h}-N,0).$$
	
	Now let us calculate the left hand side of the inequality (\ref{final2}) with updated values:
	\begin{eqnarray*}& &\hspace{-0.27 in}\sum_{p_i\in X_1} C_i'=C_0-\frac{N}{|X_M|-1}+\sum_{p_i\in X_1\setminus \lbrace p_0\rbrace} C_i-\sum_{p_i\in X_1\setminus X_m} C_{iM}\\
	& &\hspace{-0.27 in} =\sum_{p_i\in X_1} C_i-\frac{N}{|X_M|-1}-\frac{|X_M|N}{|X_M|-1}\\
	& &\hspace{-0.27 in}\ge \sum_{i=1}^n |X_i|L_i+\sum_{i=1}^h \frac{N_{d_i}-N_{d_{i+1}}}{|X_{d_i}|-1}- \frac{(|X_M|+1)N}{|X_M|-1}\\
	& &\hspace{-0.27 in} =\sum_{i=1}^n |X_i|L_i'+\frac{|X_M|N}{|X_M|-1}+\sum_{i=1}^h \frac{N_{d_i}-N_{d_{i+1}}}{|X_{d_i}|-1}- \frac{(|X_M|+1)N}{|X_M|-1}\\
	& &\hspace{-0.27 in} =\sum_{i=1}^n |X_i|L_i'+\sum_{i=1}^{h-1} \frac{(N_{d_i}-N)-(N_{d_{i+1}}-N)}{|X_{d_i}|-1}+\frac{N_{d_h}-N}{|X_{d_h}|-1}\\
	& &\hspace{-0.27 in} =\sum_{i=1}^n |X_i|L_i'+\sum_{i=1}^{h-1} \frac{N_{d_i'}'-N_{d_i'+1}'}{|X_{d_i}|-1}+\frac{N_{d_h'}'-0}{|X_{d_h}|-1}\\
	& &\hspace{-0.27 in} =\sum_{i=1}^n |X_i|L_i'+\sum_{i=1}^{h} \frac{N_{d_i'}'-N_{d_i'+1}'}{|X_{d_i}|-1}.\end{eqnarray*}
	This shows that inequality (\ref{final2}) is preserved after the procedure. Inequality (\ref{final1}) is also preserved since
	$$C_0'=C_0-\frac{N}{|X_M|-1}\ge \sum_{i=1}^n L_i-\frac{N}{|X_M|-1}=\sum_{i=1}^n L_i'.$$
	
	Now let us look at the updated version $I'$ of $I$. From (\ref{i<m}), (\ref{m<i<M}) and (\ref{i>M}) we know that $I'\subseteq I$. If $N=N_m$, then $N_m'=N_m-N=0\Rightarrow m\notin I' \Rightarrow |I'|\le |I|-1$. Similarly if $N=N_M$, then $M\notin I' \Rightarrow |I'|\le |I|-1$.
	Finally we can say that after reducing the rate of $x_M$ by $\frac{N}{|X_M|-1}$, inequalities (\ref{final1}) and (\ref{final2}) are still correct and number of inequalities from Lemma \ref{goodvec} is reduced by at least one. Hence, by the induction hypothesis, we can complete transmission of the remaining data. This completes the proof.
		\end{proof}
Combining Theorem \ref{general bound} and Theorem \ref{final} with the addition of trivial condition $C_0\ge \sum_{i=1}^n L_i$, we obtain the exact capacity region:	
	\begin{cor} \label{Corollary_1}
	Given the sets of peers $X_1,X_2,\ldots,X_n$ and upload capacities $C_0,C_1,\ldots,C_k$, the rate vector $(L_1,L_2,\ldots,L_n)$
	is achievable if and only if the following inequalities hold
				$$C_0\ge \sum_{i=1}^n L_i$$
			
				$$\sum_{p_i\in X_1} C_i \ge \sum_{i=1}^n |X_i|L_i+\sum_{i=1}^h \frac{N_{d_i}-N_{d_{i+1}}}{|X_{d_i}|-1}$$
			where $N_{d_1},\ldots,N_{d_{h+1}}$ is the dominant subsequence of $N_1,\ldots,N_{n+1}$.
Furthermore, the capacity region is achievable using routing.
	\end{cor}
%	\begin{cor}
%	For a given number of resolutions $n$ with the sets $X_1,X_2,\ldots,X_n$ and unique data segments $x_1,x_2,\ldots,x_n$ of lengths $L_1,L_2,\ldots,L_n$, let \textbf{\textit{C}} denote the capacity region of this particular network in which all the capacity vectors are achievable by routing schemes -without coding. Then, network coding does not improve this capacity region \textbf{\textit{C}}.
%	\end{cor}
%		
%	\begin{proof}
%	It follows from Corollary \ref{Corollary_1} and Theorem \ref{final}. This is because each achievable capacity vector must satisfy the inequalities in Theorem \ref{general bound}, and for every such capacity vector there is a capacity-achieving routing scheme  by Theorem \ref{final}.
%	\end{proof}	

\section{ Conclusion}

We have characterized the capacity region of achievable rates for multi-resolution streaming in peer-to-peer networks with upload capacity constraints, and shown that this region can be achieved by routing. This represents a new class of non-multicast network problems for which we have a capacity characterization. Although coding is not needed to achieve capacity in this scenario, it can nevertheless be useful in scenarios with losses or without centralized control.

%% References:
%% We recommend the usage of BibTeX:
%%
%\bibliographystyle{IEEEtran}
%\bibliography{definitions,bibliofile}
%%
%% where we here have assume the existence of the files
%% definitions.bib and bibliofile.bib.
%% BibTeX documentation can be obtained at:
%% http://www.ctan.org/tex-archive/biblio/bibtex/contrib/doc/
%%
%%
%%
%% Or manual references (pay attention to consistency!):

\section{Appendix}
		
		\textit{Proof of Lemma 1:}
		~\\
		By Induction on $k$
		
		\textit{base case:} $k=C_0$
		$$ C_0+\sum_{i=1}^k C_i = k =C_0=outdeg(v_0)\ \Rightarrow \ \sum_{i=1}^k C_i =0 \Rightarrow$$
		$$ \forall i \in \lbrace 1,\ldots,k \rbrace\ outdeg (v_i)=C_i=0 $$
		Then a directed tree rooted at $v_0$ with $k$ leaves, namely $v_1, v_2, \ldots, v_k$, gives us the desired tree.
		
		Now we can assume that $C_1\le\ldots\le C_{k-1}\le C_k$ without loss of generality. Then $C_k\ge 1$, since it would be the base case otherwise.   By inductive hypothesis on $k-1$, there exists a directed tree rooted at $v_0$ with vertices  $v_1,\ldots,v_{k-1}$ such that
				$$ outdeg(v_0) = C_0, $$
				$$ \forall i \in \lbrace 1,\ldots,k-2 \rbrace\ outdeg (v_i)=C_i, \ indeg(v_i)=1, \ and $$
				$$ outdeg (v_{k-1})=(C_{k-1}+C_k-1), \ indeg(v_{k-1})=1 $$
		since
				$$ C_0+C_1+\ldots+C_{k-2}+(C_{k-1}+C_k-1) = k-1. $$
		Now, we add a new vertex, namely $v_k$, to this tree as an isolated node and then apply the following two changes in order:
		
		1) Remove arbitrarily selected $C_k$ outgoing edges from $v_{k-1}$ where $C_k \geq 1$, then add a new directed edge from $v_{k-1}$ to $v_k$, hence we now have $outdeg (v_{k-1})=C_{k-1}$ and $indeg(v_k)=1$.
		
		2) Let $D$ be the set of nodes that are disconnected from $v_{k-1}$ in the first step. Note that $D$ is not empty since $|D|=C_k \geq 1$. Now add a new directed edge from $v_k$ to each node in set $D$. Hence, $outdeg(v_k)=C_k$.
		
		By noting that the number of outgoing/incoming edges from/to the vertices in set $D$ have not been changed, the resulting directed tree rooted at $v_0$ with vertices $v_1, v_2, \ldots, v_k$ is exactly the one which is desired. This completes the induction and the proof of the lemma.\endproof

		~\\*

		\textit{Proof of Lemma 2:}
		~\\
		  		If $C_0$, $C_i$'s or $L$ are not integers we can divide the unit length to least common multiple of denominators to make them so. Thus, we will assume that they are integers in this proof.
		  		
		  		We will do induction on $L$. Base case $L=1$ is due to lemma \ref{routree} . Now let us assume that the claim is true for $L-1$. Since
		  		$$ \sum_{i=0}^k C_i \ge kL, $$
		  		we can find integers $\overline{C_0},\overline{C_1},\ldots,\overline{C_k}$ such that
		  		$$L-1\le \overline{C_0}\le C_S-1$$
		  		$$0\le \overline{C_1} \le C_1$$
		  		$$\vdots$$
		  		$$0\le \overline{C_k} \le C_k$$
		  		
		  		$$\sum_{i=0}^{k}\overline{C_i}=k(L-1) $$
		  		
		  		If we divide the data $x$ into $L$ unit length parts $x_1,x_2,\ldots,x_L$, by the induction hypothesis we can complete the distribution of parts $x_2$ through $x_L$ by using $\overline{C_0},\overline{C_1},\ldots,\overline{C_k}$ capacities. Then we will have capacities $C_0-\overline{C_0},C_1-\overline{C_1},\ldots,C_k-\overline{C_k}$ such that
		  		$$C_0-\overline{C_0}\ge 1$$
		  		$$ \sum_{i=0}^{k}(C_i-\overline{C_i})\ge k.$$
		  		Hence, we can construct a routing tree for $x_1$ by using the Lemma \ref{routree}, thereby finishing the distribution of complete data $x=x_1 x_2\ldots x_L$.\endproof
		  		
			~\\*

				\textit{Proof of Lemma 3:}
				~\\
				We will do induction on the number of resolutions $n$.
				
				\underline{\textit{Base Case}}: $n=1$ is stated in Lemma \ref{1res}.

				\underline{\textit{Inductive Step}}:
					For $j=n$ in (\ref{gv}), we have
						\begin{eqnarray}\label{good inequality for j=n}
						\sum_{p_i\in X_n}C_i -\sum_{i=1}^{n-1}L_i\ge |X_n|L_n.
						\end{eqnarray}
					
					Since we have inequalities (\ref{gv2}) and (\ref{good inequality for j=n}), there exist $0\le \overline{C_i}\le C_i$ for each $p_i\in X_n\setminus\lbrace p_0\rbrace$ and  $L_n\le\overline{C_0}\le C_0-\sum_{i=1}^{n-1}L_i $ for source such that
						\begin{eqnarray}\label{sum C^n_i bars+M bar=X_nL_n}
						\sum_{p_i\in X_n}\overline{C_i}-\sum_{i=1}^{n-1}L_i=|X_n|L_n.
						\end{eqnarray}
					Since we have (\ref{sum C^n_i bars+M bar=X_nL_n}) and $\overline{C_0}-\sum_{i=1}^{n-1}L_i \ge L_n$, by Lemma \ref{1res}, each peer in set $X_n$ can receive the resolution $x_n$ completely by using $\overline{C_i}$ capacity from each $p_i\in X_n$ including source. After completing the transmission of the resolution $x_n$, each peer $p_i \in X_n$ needs to receive only the resolutions $x_1, x_2, \ldots, x_{n-1}$ same as all the peers in $X_{n-1}$.
					
					If we denote the remaining capacity of a peer $p_i$ with $C_i'$ and of source with $C_0'$,  we may write
					$$
							 C_i'=
							\left\{
							\begin{array}{cl}
							 C_i & if \ p_i\in X_1\setminus X_n \\
							 C_i-\overline{C_i} & if \ p_i \in X_n \\
			
							\end{array}
							\right.
							$$
					Then, for each $j\in\lbrace 1,\ldots,n-1\rbrace$, following is true by using (\ref{gv}) and (\ref{sum C^n_i bars+M bar=X_nL_n}):
					$$\sum_{p_i\in X_j}C_i'-\sum_{i=1}^{j-1}L_i=-\sum_{i=1}^{j-1}L_i+\sum_{p_i\in X_j\setminus X_n}C_i+\sum_{p_i\in X_n}(C_i-\overline{C_i})$$
					$$=-\sum_{i=1}^{j-1}L_i+\sum_{p_i\in X_j}C_i-|X_n|L_n$$
					$$\ge\sum_{i=j}^n |X_i|L_i-|X_n|L_n=\sum_{i=j}^{n-1} |X_i|L_i.$$
					Hence we have, for resolutions $x_1,\ldots,x_{n-1}$, for every $j\in\lbrace 1,\ldots,n-1\rbrace$
					$$\sum_{p_i\in X_j}C_i'\ge\sum_{i=1}^{j-1}L_i + \sum_{i=j}^{n-1} |X_i|L_i$$
					and
					$$	C_0'\ge\sum_{i=1}^{n-1} L_i.$$
					Finally we can claim that these capacities are enough for transmission of $x_1,\ldots,x_{n-1}$ by the inductive hypothesis, thereby completing all transmission.\endproof

		~\\*
				\textit{Proof of Lemma 4:}
				
			 	For simplicity, we use the notation
			 	 $$H(S)=H(X_1\rightarrow S)$$
			 	to represent the entropy of all information $X_1\rightarrow S$ recieved by set $S$ and use analogous notation for other information theoretic quantities. Then
			 	$$\sum_{j=1}^{n}I(X_j\rightarrow (X_1\setminus X_j);x_j)\le\sum_{j=1}^{n}I(X_1\setminus X_j;x_j)= \sum_{j=1}^{n}\sum_{i=1}^{j-1} I(X_i\setminus X_{i+1};x_j|X_{i+1}\setminus X_{j})$$
			 	$$=\sum_{i=1}^{n-1}\sum_{j=i+1}^{n}I(X_i\setminus X_{i+1};x_j|X_{i+1}\setminus X_{j})$$
			 	Let us examine the inner sum:
			 	$$\sum_{j=i+1}^{n}I(X_i\setminus X_{i+1};x_j|X_{i+1}\setminus X_{j})=H(X_i\setminus X_{i+1})-H(X_i\setminus X_{i+1}|x_{i+1})+ \sum_{j=i+2}^{n}I(X_i\setminus X_{i+1};x_j|X_{i+1}\setminus X_{j})$$
			 	$$=H(X_i\setminus X_{i+1}|x_1,\ldots,x_i)+I(X_i\setminus X_{i+1};x_1,\ldots,x_i) -H(X_i\setminus X_{i+1}|x_1,\ldots,x_{i+1})-I(X_i\setminus X_{i+1};x_1,\ldots,x_i|x_{i+1})$$
			 	$$+\sum_{j=i+2}^{n}I(X_i\setminus X_{i+1};x_j|X_{i+1}\setminus X_{j})$$
			 	\textit{Since $x_1,\ldots,x_i$ are independent from $x_{i+1}$ and peers in set $X_i\setminus X_{i+1}$ can decode $x_1,\ldots,x_i$:}
			 	$$=H(X_i\setminus X_{i+1}|x_1,\ldots,x_i)+H(x_1,\ldots,x_i) -H(X_i\setminus X_{i+1}|x_1,\ldots,x_{i+1})-H(x_1,\ldots,x_i)$$
			 	$$+\sum_{j=i+2}^{n}I(X_i\setminus X_{i+1};x_j|X_{i+1}\setminus X_{j})$$
			 	$$= H(X_i\setminus X_{i+1}|x_1,\ldots,x_i) -H(X_i\setminus X_{i+1}|x_1,\ldots,x_{i+1})+\sum_{j=i+2}^{n}I(X_i\setminus X_{i+1};x_j|X_{i+1}\setminus X_{j})$$
			 	$$=H(X_i\setminus X_{i+1}|x_1,\ldots,x_i) -H(X_i\setminus X_{i+1}|x_1,\ldots,x_{i+1})+ \sum_{j=i+2}^{n}[H(X_i\setminus X_{i+1}|X_{i+1}\setminus X_{j})-H(X_i\setminus X_{i+1}|x_j,X_{i+1}\setminus X_{j})] $$
			 	$$=H(X_i\setminus X_{i+1}|x_1,\ldots,x_i)-H(X_i\setminus X_{i+1}|x_1,\ldots,x_{i+1})+H(X_i\setminus X_{i+1}|X_{i+1}\setminus X_{i+2})$$
			 	$$-\sum_{j=i+2}^{n}[H(X_i\setminus X_{i+1}|x_j,X_{i+1}\setminus X_{j})-H(X_i\setminus X_{i+1}|X_{i+1}\setminus X_{j+1})]-H(X_i\setminus X_{i+1}|x_n,X_{i+1}\setminus X_{n})$$
			 	$$\le H(X_i\setminus X_{i+1}|x_1,\ldots,x_i)\le \sum_{p_t\in X_i\setminus X_{i+1}}H(p_t|x_1,\ldots,x_i)=\sum_{p_t\in X_i\setminus X_{i+1}}[H(p_t)-I(p_t;x_1,\ldots,x_i)]$$
			 	$$=\sum_{p_t\in X_i\setminus X_{i+1}}H(p_t)-|X_i\setminus X_{i+1}|\sum\limits_{j=1}^{i}H(x_j)$$
			 	Hence we can write
			 	$$\sum_{j=1}^{n}I(X_j\rightarrow (X_1\setminus X_j);x_j)\le \sum_{i=1}^{n-1}\sum_{j=i+1}^{n}I(X_i\setminus X_{i+1};x_j|X_{i+1}\setminus X_{j})$$
			 	$$\le \sum_{i=1}^{n-1}\left(\sum_{p_t\in X_i\setminus X_{i+1}}H(p_t)-|X_i\setminus X_{i+1}|\sum\limits_{j=1}^{i}H(x_j)\right)=\sum_{p_i \in X_1\setminus \lbrace p_0 \rbrace}H(p_i) - \sum_{j=1}^{n} |X_j|H(x_j)$$
			 	$$=\sum_{p_i \in X_1\setminus \lbrace p_0 \rbrace}H(X_1\rightarrow p_i) - \sum_{j=1}^{n} |X_j|H(x_j)\le \sum_{p_i \in X_1\setminus \lbrace p_0 \rbrace}\sum_{p_l\in X_1}H(p_l\rightarrow p_i) - \sum_{j=1}^{n} |X_j|H(x_j)$$
			 	$$=\sum_{p_l\in X_1}\sum_{p_i \in X_1\setminus \lbrace p_0 \rbrace}H(p_l\rightarrow p_i) - \sum_{j=1}^{n} |X_j|H(x_j)\le \sum_{p_l\in X_1}C_l - \sum_{j=1}^{n} |X_j|H(x_j) $$
			 	Last inequality is due to \textit{property iii}.\endproof
		
		~\\*
		\textit{Proof of Lemma 5:}
						~\\
			 		$$A_j\ge N_j\Leftrightarrow$$
			 		$$\sum_{l=j}^n|X_l|H(x_l)-\sum_{l=j}^{n}\left(I(X_l\rightarrow (X_1\setminus X_l);x_l)+\sum_{p_i\in X_l}I(X_l\rightarrow p_i;x_l|X_l\rightarrow (X_1\setminus X_l))\right)$$
			 		$$\ge \sum_{l=j}^n |X_l|L_l+\sum_{l=1}^{j-1}L_l-\sum_{p_i\in X_j}C_i.$$
			 		 Then, all we have to prove is , putting $H(x_l)$ in place of $L_l$,
			 		
			 		$$\sum_{p_i\in X_j} C_i\ge\sum_{l=1}^{j-1}H(x_l) + \sum_{l=j}^{n} I(X_l\rightarrow (X_1\setminus X_l);x_l)+\sum_{l=j}^{n} \sum_{p_i\in X_l}I(X_l\rightarrow p_i;x_l|X_l\rightarrow (X_1\setminus X_l))$$
			 		~\\
			 		Following lemma will help us
			 		\begin{sublemma}
			 		For each $j$,$m\ge j$ and variables $A$,$B$
			 		$$\sum_{l=j}^{m}I(A;x_l|X_l\rightarrow X_1\setminus X_l,B)\le H(A).$$
			 		\end{sublemma}
			 			\begin{proof}
			 			$$\sum_{l=j}^{m}I(A;x_l|X_l\rightarrow X_1\setminus X_l,B)=$$
			 			$$\sum_{l=j}^{m}H(A|X_l\rightarrow X_1\setminus X_l,B)-\sum_{l=j}^{m} H(A|x_l,X_l\rightarrow X_1\setminus X_l,B)$$
			 			$$=H(A|X_j\rightarrow X_1\setminus X_j,B)+\sum_{l=j+1}^{m}H(A|X_l\rightarrow X_1\setminus X_l,B)$$
			 			$$-\sum_{l=j}^{m-1} H(A|x_l,X_l\rightarrow X_1\setminus X_l,B)-H(A|x_l,X_m\rightarrow X_1\setminus X_m,B)$$
			 			$$=H(A|X_j\rightarrow X_1\setminus X_j,B)+\sum_{l=j}^{m-1}H(A|X_{l+1}\rightarrow X_1\setminus X_{l+1},B)$$
			 			$$-\sum_{l=j}^{m-1} H(A|x_l,X_l\rightarrow X_1\setminus X_l,B)-H(A|x_l,X_m\rightarrow X_1\setminus X_m,B)$$
			 			$$\le H(A)+\sum_{l=j}^{m-1}[H(A|X_{l+1}\rightarrow X_1\setminus X_{l+1},B)- H(A|x_l,X_l\rightarrow X_1\setminus X_l,B)].$$
			 			Let us examine $X_l\rightarrow X_1\setminus X_l$:
			 			$$X_1\rightarrow X_1\setminus X_l=(X_l\rightarrow X_1\setminus X_l)\cup (X_1\setminus X_l\rightarrow X_1\setminus X_l).$$
			 			With a similar idea used in the proof of Theorem 2, we can claim that $X_1\setminus X_l\rightarrow X_1\setminus X_l$ is completely dependent on $X_l\rightarrow X_1\setminus X_l$. Thus $X_1\rightarrow X_1\setminus X_l$ is also completely dependent on $X_l\rightarrow X_1\setminus X_l$. Since $ X_l\rightarrow X_1\setminus X_l\subseteq X_1\rightarrow X_1\setminus X_l$, reverse is also true. Hence, we can use sets $ X_l\rightarrow X_1\setminus X_l$ and $X_1\rightarrow X_1\setminus X_l$ interchangeably in information theoretic expressions. Then,
			 			$$\sum_{l=j}^{m}I(A;x_l|X_l\rightarrow X_1\setminus X_l,B)\le$$
			 			$$H(A)+\sum_{l=j}^{m-1}[H(A|X_{l+1}\rightarrow X_1\setminus X_{l+1},B)- H(A|x_l,X_l\rightarrow X_1\setminus X_l,B)]=$$
			 			$$H(A)+\sum_{l=j}^{m-1}[H(A|X_1\rightarrow X_1\setminus X_{l+1},B)- H(A|x_l,X_1\rightarrow X_1\setminus X_l,B)]=$$
			 			$$H(A)+\sum_{l=j}^{m-1}[H(A|X_1\rightarrow X_l\setminus X_{l+1},X_1\rightarrow X_1\setminus X_l,B)- H(A|x_l,X_1\rightarrow X_1\setminus X_l,B)]=$$
			 			$$H(A)-\sum_{l=j}^{m-1}I(A;X_1\rightarrow X_l\setminus X_{l+1}|x_l,X_1\rightarrow X_1\setminus X_l,B)\le H(A).$$
			 						 			\end{proof}
			 		Now, let us define $r(p_i)$ such that $p_i\in X_{r(p_i)}\setminus  X_{r(p_i)+1}$.Then,
			 		$$\sum_{l=j}^{n}\sum_{p_i\in X_l}I(X_l\rightarrow p_i;x_l|X_l\rightarrow (X_1\setminus X_l))=
			 		  \sum_{p_i\in X_j}\sum_{l=j}^{r(p_i)}I(X_l\rightarrow p_i;x_l|X_l\rightarrow (X_1\setminus X_l))=$$
			 		$$\sum_{p_i\in X_j}\sum_{l=j}^{r(p_i)} I\left(\bigcup_{m=l}^{r(p_i)-1}\left((X_m\setminus X_{m+1})\rightarrow p_i\right)\cup (X_{r(p_i)}\rightarrow p_i);x_l|X_l\rightarrow X_1\setminus X_l\right)$$
			 		By the chain rule
			 		$$=\sum_{p_i\in X_j}\sum_{l=j}^{r(p_i)}\left( \sum_{m=l}^{r(p_i)-1}I(X_m\setminus X_{m+1}\rightarrow p_i;x_l|X_l\rightarrow X_1\setminus X_l,X_{m+1}\rightarrow p_i)+I((X_{r(p_i)}\rightarrow p_i);x_l|X_l\rightarrow X_1\setminus X_l) \right)= $$
			 		$$\sum_{p_i\in X_j}\left(\sum_{m=j}^{r(p_i)-1}\sum_{l=j}^{m}I(X_m\setminus X_{m+1}\rightarrow p_i;x_l|X_l\rightarrow X_1\setminus X_l,X_{m+1}\rightarrow p_i)+\sum_{l=j}^{r(p_i)}I(X_{r(p_i)}\rightarrow p_i;x_l|X_l\rightarrow X_1\setminus X_l)\right)$$
			 		by Sublemma 1( we also seperate $p_1\in X_n$ from other peers, and writing it in close form)
			 		$$\le \sum_{l=j}^{n}I(X_l\rightarrow p_1;x_l|X_l\rightarrow X_1\setminus X_l)+ \sum_{p_i\in X_j\setminus\lbrace p_1\rbrace }\left(\sum_{m=j}^{r(p_i)-1}H(X_m\setminus X_{m+1}\rightarrow p_i)+H(X_{r(p_i)}\rightarrow p_i)\right)$$
			 		Then
			 		$$\sum_{l=1}^{j-1}H(x_l) + \sum_{l=j}^{n} I(X_l\rightarrow (X_1\setminus X_l);x_l)+\sum_{l=j}^{n} \sum_{p_i\in X_l}I(X_l\rightarrow p_i;x_l|X_l\rightarrow (X_1\setminus X_l))\le$$
			 ~\\*		
			 		$$\sum_{l=1}^{j-1}H(x_l) + \sum_{l=j}^{n} I(X_l\rightarrow (X_1\setminus X_l);x_l)+\sum_{l=j}^{n}I(X_l\rightarrow p_1;x_l|X_l\rightarrow X_1\setminus X_l)+$$
			 		$$ \sum_{p_i\in X_j\setminus\lbrace p_1\rbrace }\left(\sum_{m=j}^{r(p_i)-1}H(X_m\setminus X_{m+1}\rightarrow p_i)+H(X_{r(p_i)}\rightarrow p_i)\right)=$$
			 ~\\*
			 ~\\*
			 		$$\sum_{l=1}^{j-1}H(x_l) + \sum_{l=j}^{n} [I(X_l\rightarrow (X_1\setminus X_l);x_l)+I(X_l\rightarrow p_1;x_l|X_l\rightarrow X_1\setminus X_l)]+$$
			 		$$ \sum_{p_i\in X_j\setminus\lbrace p_1\rbrace }\left(\sum_{m=j}^{r(p_i)-1}H(X_m\setminus X_{m+1}\rightarrow p_i)+H(X_{r(p_i)}\rightarrow p_i)\right)=$$
			 ~\\*
			 ~\\*
			 		$$\sum_{l=1}^{j-1}H(x_l) + \sum_{l=j}^{n} I(X_l\rightarrow X_1\setminus X_l,X_l\rightarrow p_1;x_l)+$$
			 		$$ \sum_{p_i\in X_j\setminus\lbrace p_1\rbrace }\left(\sum_{m=j}^{r(p_i)-1}H(X_m\setminus X_{m+1}\rightarrow p_i)+H(X_{r(p_i)}\rightarrow p_i)\right)=$$		
			 ~\\*
			 ~\\*
			 		$$\sum_{l=1}^{j-1}H(x_j) + \sum_{l=j}^{n} H(x_l)+ \sum_{p_i\in X_j\setminus\lbrace p_1\rbrace }\left(\sum_{m=j}^{r(p_i-1}H(X_m\setminus X_{m+1}\rightarrow p_i)+H(X_{r(p_i)}\rightarrow p_i)\right)=$$
			 ~\\*
			 ~\\*
			 		$$\sum_{l=1}^{n}H(x_l)+\sum_{p_i\in X_j\setminus\lbrace p_1\rbrace }\left(\sum_{m=j}^{r(p_i)-1}H(X_m\setminus X_{m+1}\rightarrow p_i)+H(X_{r(p_i)}\rightarrow p_i)\right)=$$
			 		$$H(X_1\rightarrow p_1)+\sum_{p_i\in X_j\setminus\lbrace p_1\rbrace }\left(\sum_{m=j}^{r(p_i)-1}H(X_m\setminus X_{m+1}\rightarrow p_i)+H(X_{r(p_i)}\rightarrow p_i)\right)\le$$
			 		$$H(X_j\rightarrow p_1)+H(X_1\setminus X_j\rightarrow p_1)+\sum_{p_i\in X_j\setminus\lbrace p_1\rbrace }\left(\sum_{m=j}^{r(p_i)-1}H(X_m\setminus X_{m+1}\rightarrow p_i)+H(X_{r(p_i)}\rightarrow p_i)\right)\le$$
			 		$$\sum_{p_l\in X_j}H(p_l\rightarrow p_1)+\sum_{p_l\in X_j}H(p_l\rightarrow X_1\setminus X_j)+\sum_{p_l\in X_j}\sum_{p_i\in X_j\setminus\lbrace p_1\rbrace }H(p_l\rightarrow p_i)\le$$
			 		$$\sum_{p_l\in X_j}\sum_{p_i\in X_1\setminus X_j}H(p_l\rightarrow p_i)+\sum_{p_l\in X_j}\sum_{p_i\in X_j}H(p_l\rightarrow p_i)=\sum_{p_l\in X_j}\sum_{p_i\in X_1}H(p_l\rightarrow p_i) \le \sum_{p_l\in X_j}C_l$$
			 		Last inequality is due to \textit{property iii}.\endproof
			 		
			 	~\\*	
		\textit{Proof of Lemma 6:}
								~\\
			Let $A_1=\{N_1, N_2, \ldots, N_{m-1}\}$ and $j_1= min\{r: N_r= max(A_1)\}$. Then, iteratively define $A_i=A_{i-1} \setminus \{N_{j_{i-1}}, N_{j_{i-1}+1}, \ldots, N_{m-1}\}$ and $j_i=min\{r: N_r= max(A_i)\}$ for $i>1$. Note that $\exists \ l \in \{1, 2, \ldots, m-1\}$ such that $j_l=1$ since $|A_i|$ is strictly decreasing and so is the sequence $\{j_i\}$. Moreover, note that $m>j_1>j_2>\ldots>j_l=1$ and $0 \ge N_{j_1}>N_{j_2}>\ldots>N_{j_l}=N_1$.
				\begin{sublemma}\label{sub_1}
				$$N-N_1 \ge \frac{|X_M|N}{|X_M|-1}$$
				\end{sublemma}
				\begin{proof}
				Using $\sum\limits_{p_i \in X_1}C_i+C_s- \sum\limits_{i=1}^{n}|X_i|L_i=-N_1$, we get
					\begin{eqnarray*}
					\frac{N}{|X_M|-1} &\le& \frac{N_M}{|X_M|-1} \\
									  &\le& \sum_{j=1}^{h-1} \frac{N_{i_j}-N_{i_{j+1}}}{|X_{i_j}|-1}+\frac{N_M}{|X_M|-1} \ \\
									  &=& \sum_{j=1}^{h} \frac{N_{i_j}-N_{i_{j+1}}}{|X_{i_j}|-1} \ since \ i_h=M \ and \ N_{i_{h+1}}=0 \\
									  &\le& \sum_{p_j \in X_1}C_j+C_s-\sum_{j=1}^{n}|X_j|L_j \\
									  &=& -N_1
					\end{eqnarray*}
				Hence, we get 	$\frac{N}{|X_M|-1} \le -N_1$ which implies $N-N_1 \ge \frac{|X_M|N}{|X_M|-1}$, as desired.
				\end{proof}
				\begin{identity}\label{sub_2}
				If $j<t$,
				$$\sum_{p_i \in X_j \setminus X_t}C_i=N_t-N_j+\sum_{i=j}^{t-1}L_i(|X_i|-1)$$ holds.
				\end{identity}
				\begin{proof}
				Note that we have $\sum_{p_i \in X_j}C_i=\sum_{i=j}^{n}|X_i|L_i+\sum_{i=1}^{j-1}L_i-C_s-N_j$, $\sum_{p_i \in X_t}C_i=\sum_{i=t}^{n}|X_i|L_i+\sum_{i=1}^{t-1}L_i-C_s-N_t$ and $\sum_{p_i \in X_j \setminus X_t}C_i = \sum_{p_i \in X_j}C_i-\sum_{p_i \in X_t}C_i$. Hence, we get
					\begin{eqnarray*}
					\sum_{p_i \in X_j \setminus X_t}C_i &=& (\sum_{i=j}^{n}|X_i|L_i+\sum_{i=1}^{j-1}L_i-C_s-N_j)-		 (\sum_{i=t}^{n}|X_i|L_i+\sum_{i=1}^{t-1}L_i-C_s-N_t)\\
					&=&N_t-N_j+\sum_{i=j}^{t-1}L_i(|X_i|-1).
					\end{eqnarray*}
					This completes the proof of identity.
				\end{proof}
				
				We will give the proof of lemma in the following two complementary cases:
				\linebreak
				\linebreak
				\underline{\textit{Case 1}}: $\frac{|X_M|N}{|X_M|-1} \le N-N_{j_1}$  holds.
						
					By Identity \ref{sub_2}, we have
					\begin{eqnarray}\label{ineq}
					\sum_{p_i \in X_{j_1} \setminus X_m}C_i &=& N_m-N_{j_1}+\sum_{i=j_1}^{m-1}L_i(|x_i|-1) \notag\\
					&\ge& N-N_{j_1} \\
					&\ge& \frac{|X_M|N}{|X_M|-1} \notag
					\end{eqnarray}
					Hence, $\forall p_i \in X_{j_1}\setminus X_m \ \exists 0 \le C_{iM} \le C_i $ such that
						\begin{eqnarray}\label{eqn_1}
						\sum_{p_i \in X_{j_1} \setminus X_m}C_{iM}=\frac{|X_M|N}{|X_M|-1}.
						\end{eqnarray}
						Note that
						\begin{eqnarray}\label{ineq_1}
						\sum_{p_i \in X_j \setminus X_m}C_{iM} \le \sum_{p_i \in X_{j_1} \setminus X_m}C_{iM}=\frac{|X_M|N}{|X_M|-1} \le N-N_{j_1} \le N-N_j
						\end{eqnarray}
					$\forall j \in \{j_1, j_1+1,\ldots,m-1\}$ since $N_{j_1} \ge N_j \ \forall j \in \{j_1, j_1+1,\ldots,m-1\}$ by the choice of $j_1$. Now, we choose $C_{iM}=0, \ \forall p_i \in X_1 \setminus X_{j_1}$. Hence, usign $C_{iM}=0, \ \forall p_i \in X_1 \setminus X_{j_1}$ we get,
						\begin{eqnarray} \label{ineq_2}
						\sum_{p_i \in X_j \setminus X_m}C_{iM} = \sum_{p_i \in X_{j_1} \setminus X_m}C_{iM}=\frac{|X_M|N}{|X_M|-1} \le N-N_{j_1} \le N-N_j
						\end{eqnarray}
					$\forall j \in \{1, 2,\ldots,j_1-1\}$ since $N_{j_1} \ge N_j \ \forall j \in \{1, 2,\ldots,j_1-1\}$ by the choice of $j_1$. Hence, as can be seen from (\ref{eqn_1}), (\ref{ineq_1}) and (\ref{ineq_2}), this particular choice of $C_{iM} \ \forall p_i \in X_1 \setminus X_m$ is exactly as desired. This completes proof of \textit{Case 1}. \linebreak

				\underline{\textit{Case 2}}: $\frac{|X_M|N}{|X_M|-1} > N-N_{j_1}$  holds.
				
					Substituting $N_1=N_{j_l}$ in Sublemma, we get:
					$$N-N_{j_l} \ge \frac{|X_M|N}{|X_M|-1} > N-N_{j_1}$$ which implies that $\exists t \in \{1, 2,\ldots, l\}$ such that
					\begin{eqnarray}\label{choice of t}
					N-N_{j_{t+1}} \ge \frac{|X_M|N}{|X_M|-1} > N-N_{j_t}
					\end{eqnarray}
					holds. Now, we will give an iterative method for the choice of $C_{iM}, \ \forall p_i \in X_1\setminus X_m$ satisfying required constraints as follows:
					
					\textit{\underline{First Step}}: Choice of $C_{iM}, \ \forall p_i \in X_{j_1}\setminus X_m$
					
						By (\ref{ineq}), we have {$\sum_{p_i \in X_{j_1} \setminus X_m}C_i \ge N-N_{j_1}$}. Hence,
						$\forall p_i \in X_{j_1}\setminus X_m \ \exists C_{iM}$ with $0 \le C_{iM} \le C_i $ such that
							\begin{eqnarray}\label{eqn_2}
							\sum_{p_i \in X_{j_1} \setminus X_m}C_{iM}=N-N_{j_1}.
							\end{eqnarray}
						Note that	
							\begin{eqnarray}\label{ineq_3}
							\sum_{p_i \in X_j \setminus X_m}C_{iM} \le \sum_{p_i \in X_{j_1} \setminus X_m}C_{iM}=N-N_{j_1} \le N-N_j
							\end{eqnarray}
							$\forall j \in \{j_1, j_1+1,\ldots,m-1\}$ since $N_{j_1} \ge N_j \ \forall j \in \{j_1, j_1+1,\ldots,m-1\}$ by the choice of $j_1$.
					
					\textit{\underline{$r$-th step for $1<r\le t$ }}: Choice of $C_{iM}, \ \forall p_i \in X_{j_r}\setminus X_{j_{r-1}}$.
					
						By Identity \ref{sub_2}, we have
						\begin{eqnarray*}
						\sum_{p_i \in X_{j_r} \setminus X_{j_{r-1}}}C_i &=& N_{j_{r-1}}-N_{j_r}+\sum_{i=j_r}^{j_{r-1}}L_i(|x_i|-1) \\
										&\ge& N_{j_{r-1}}-N_{j_r}
						\end{eqnarray*}
						Hence,
						$\forall p_i \in X_{j_r}\setminus X_{j_{r-1}} \ \exists C_{iM}$ with $0 \le C_{iM} \le C_i $ such that
						\begin{eqnarray}\label{eqn_3}
						\sum_{p_i \in X_{j_r} \setminus X_{j_{r-1}}}C_{iM}=N_{j_{r-1}}-N_{j_r}.
						\end{eqnarray}
						Using (\ref{eqn_2}) and (\ref{eqn_3}), $\forall r \in \{2, 3,\ldots, t\}$ and $\forall j \in \{j_r, j_r+1,\ldots,j_{r-1}+1\}$ we have
							\begin{eqnarray}\label{ineq_4}
							\sum_{p_i \in X_{j} \setminus X_m}C_{iM} &\le& \sum_{p_i \in X_{j_r} \setminus X_m}C_{iM} \notag \\
							&=& \sum_{p_i \in X_{j_1} \setminus X_m}C_{iM}+\sum_{q=2}^{r}(\sum_{p_i \in X_{j_q} \setminus X_{j_{q-1}}}C_{iM})  \notag \\
							&=& (N-N_{j_1})+\sum_{q=2}^{r}(N_{j_{q-1}}-N_{j_q}) \notag \\
							&=& N-N_{j_r}  \notag \\
							&\le& N-N_j
							\end{eqnarray}
						since $N_{j_r}=max\{N_1, N_2, \ldots, N_{j_{r-1}+1}\}$ by the choice of $j_r$.
					
					\textit{\underline{$(t+1)$-th step}}: Choice of $C_{iM}, \ \forall p_i \in X_{j_{t+1}}\setminus X_{j_t}$ and $\forall p_i \in X_1\setminus X_{j_{t+1}}$, separately.
						
						Remember we have (\ref{choice of t}) by the choice of $t$. Using Identity \ref{sub_2} and (\ref{choice of t}), we get
						\begin{eqnarray*}
						\sum_{p_i \in X_{j_{t+1}} \setminus X_{j_t}}C_{iM} &\ge& N_{j_t}-N_{j_{t+1}} \\
						&=& (N-N_{j_{t+1}})-(N-N_{j_t})\\
						&\ge& \frac{|X_M|N}{|X_M|-1}-(N-N_{j_t})
						\end{eqnarray*}
						Hence $\forall p_i \in X_{j_{t+1}}\setminus X_{j_t}, \ \exists C_{iM}$ with $0 \le C_{iM} \le C_i $ such that
						\begin{eqnarray}\label{eqn_4}
						\sum_{p_i \in X_{j_{t+1}} \setminus X_{j_t}}C_{iM}=\frac{|X_M|N}{|X_M|-1}-(N-N_{j_t}).
						\end{eqnarray}
						Using (\ref{eqn_4}), we get
						\begin{eqnarray}\label{eqn_5}
						\sum_{p_i \in X_{j_{t+1}} \setminus X_m}C_{iM} &=& 	\sum_{p_i \in X_{j_{t+1}} \setminus X_{j_t}}C_{iM}+	\sum_{p_i \in X_{j_t} \setminus X_m}C_{iM} \notag \\
						&=&\frac{|X_M|N}{|X_M|-1}-(N-N_{j_t})+(N-N_{j_t}) \notag \\
						&=& \frac{|X_M|N}{|X_M|-1}
						\end{eqnarray}
						Now, choose $C_{iM}=0$, $\forall p_i \in X_1\setminus X_{j_{t+1}}$. Using (\ref{eqn_5}) and the choice of $C_{iM}=0$, $\forall p_i \in X_1\setminus X_{j_{t+1}}$, we get for any $j \in \{1, 2, \ldots, j_{t+1}-1\}$ that
						\begin{eqnarray}\label{eqn_6}
						\sum_{p_i \in X_j \setminus X_m}C_{iM}&=&\sum_{p_i \in  X_j\setminus X_{j_{t+1}}}C_{iM}+\sum_{p_i \in X_{j_{t+1}} \setminus X_m}C_{iM} \notag \\
						&=&0+\frac{|X_M|N}{|X_M|-1} \notag \\
						&=&\frac{|X_M|N}{|X_M|-1}
						\end{eqnarray}
						Using (\ref{eqn_6}), we get for any $j \in \{1, 2, \ldots, j_t-1\}$ that
						\begin{eqnarray}\label{ineq_5}
						\sum_{p_i \in X_j \setminus X_m}C_{iM} &\le& \sum_{p_i \in X_{j_{t+1}} \setminus X_m}C_{iM} \notag \\
						&=&\frac{|X_M|N}{|X_M|-1} \notag \\
						&\le& N-N{j_{t+1}} \notag \\
						&\le& N-N_j
						\end{eqnarray}
						since $N_{j_{t+1}}=max\{N_1, N_2, \ldots, N_{j_t-1}\}$ by the choice of $j_{t+1}$. Putting $j=1$ in (\ref{eqn_6}) gives
						\begin{eqnarray}\label{eqn_7}
						\sum_{p_i \in X_1 \setminus X_m}C_{iM}=\frac{|X_M|N}{|X_M|-1}
						\end{eqnarray}
						Finally, considering (\ref{ineq_3}), (\ref{ineq_4}) and (\ref{ineq_5}), this particular choice of $C_{iM}$ for each peer $p_i \in X_1 \setminus X_m$ satisfies
						$$\sum_{p_i \in X_j\setminus X_m} C_{iM} \le N-N_j$$ $\forall j\in\{1, 2, \ldots, m-1\}$ and by (\ref{eqn_7})
						$$\sum_{p_i \in X_1\setminus X_m} C_{iM} = \frac{|X_M|N}{|X_M|-1},$$ as desired. This completes the proof of Lemma \ref{Choice of C_iM}.\endproof

\end{document}